\documentclass{tlp}

\usepackage{amsmath}
\usepackage{graphicx}
\usepackage{multirow}
\usepackage{mathtools}
\usepackage[hidelinks]{hyperref}
\usepackage{flexisym}
\usepackage{xcolor}
\usepackage{enumitem}
\usepackage{xspace}
\usepackage{cleveref}
\usepackage{booktabs}
\usepackage{array}
\usepackage{tikz}
\usepackage{scalerel,amssymb}
\usepackage{subcaption}
\usepackage{thm-restate}
\usetikzlibrary{positioning, fit, calc, shapes.geometric}
\tikzset{block/.style={draw, thick, text width=2cm , minimum height=1.3cm, align=center},   
line/.style={-latex}     
} 

\newcommand{\up}[2]{\ensuremath{{#1}|_{{#2}}}}
\newcommand{\Card}[1]{|#1|}
\newcommand{\var}[1]{\mathsf{Var}(#1)}
\newcommand{\copyop}[1]{\mathsf{Copy}(#1)}

\newcommand{\loopatoms}[1]{\mathsf{LA}(#1)}
\newcommand{\completion}[1]{\mathsf{Comp}(#1)}

\newcommand{\answer}[1]{\mathsf{AS}(#1)}

\newcommand{\copyatom}[1]{#1\textprime}
\newcommand{\staratom}[1]{#1^{\star}}
\newcommand{\copyvar}{\mathsf{CV}}
\newcommand{\at}[1]{\mathsf{at}(#1)}
\newcommand{\rules}[1]{#1}
\newcommand{\true}{\textit{true}\xspace}
\newcommand{\false}{\textit{false}\xspace}

\newcommand{\head}[1]{\mathsf{Head}(#1)}
\newcommand{\body}[1]{\mathsf{Body}(#1)}

\newcommand{\toolname}{\ensuremath{\mathsf{sharpASP}}-\ensuremath{\mathcal{SR}}\xspace}
\newcommand{\ganak}{GANAK}
\newcommand{\clingo}{clingo\xspace}
\newcommand{\dependency}[1]{\mathsf{DG}(#1)}
\newcommand{\funcname}[1]{\psi(#1)}

\newcommand{\coNP}{$\mathsf{co}$-$\mathsf{NP}$\xspace}
\newcommand{\pclass}{$\mathsf{P}$\xspace}
\newcommand{\npclass}{$\mathsf{NP}$\xspace}
\newcommand{\first}{\ensuremath{\phi_1}}
\newcommand{\second}{\ensuremath{\phi_2}}
\newcommand{\project}{\ensuremath{\mathcal{X}}}

\newtheorem{lemma}{Lemma}
\newtheorem{example}{Example}
\newtheorem{proposition}{Proposition}

\begin{document}

\lefttitle{Cambridge Author}

\jnlPage{1}{8}
\jnlDoiYr{2021}
\doival{10.1017/xxxxx}

\title[Theory and Practice of Logic Programming]{Counting Answer Sets of Disjunctive Answer Set Programs}

\begin{authgrp}
\author{Md Mohimenul Kabir}
\affiliation{National University of Singapore}
\author{Supratik Chakraborty}
\affiliation{Indian Institute of Technology Bombay}
\author{Kuldeep S Meel}
\affiliation{Georgia Institute of Technology}
\end{authgrp}
\history{\sub{xx xx xxxx;} \rev{xx xx xxxx;} \acc{xx xx xxxx}}

\maketitle

\begin{abstract}
Answer Set Programming (ASP) provides a powerful declarative paradigm for 
  knowledge representation and reasoning. Recently, counting answer sets 
  has emerged as an important computational problem with applications in 
  probabilistic reasoning, network reliability analysis, and other domains. 
  This has motivated significant research into designing efficient ASP 
  counters. While substantial progress has been made for normal logic 
  programs, the development of practical counters for disjunctive logic 
  programs  remains challenging.
  
  We present {\toolname}, a novel framework for counting answer sets of 
  disjunctive logic programs based on subtractive reduction to projected 
  propositional model counting. Our approach introduces an alternative 
  characterization of answer sets that enables efficient reduction while ensuring that intermediate representations remain of polynomial size. 
  This allows 
  {\toolname} to leverage recent advances in projected model counting 
  technology. Through extensive experimental evaluation on diverse 
  benchmarks, we demonstrate that {\toolname} significantly outperforms 
  existing counters on instances with large answer set counts. Building on these results, 
  we develop a hybrid counting approach that combines 
  enumeration techniques with {\toolname} to achieve state-of-the-art 
  performance across the full spectrum of disjunctive programs.
\end{abstract}

\begin{keywords}
Answer Set Counting, Disjunctive Programs, Subtractive Reduction, Projected Model Counting
\end{keywords}
\section{Introduction}
Answer Set Programming (ASP) \citep{MT1999} has emerged as a powerful 
declarative problem-solving paradigm with applications across 
diverse application domains. These include decision support systems 
\citep{NBGWB2001}, systems biology \citep{GSTUV2008}, diagnosis and 
repair \citep{LR2015}. In 
the ASP paradigm, domain knowledge and queries are expressed through 
rules defined over propositional atoms, collectively forming an ASP 
program. Solutions manifest as answer sets - assignments to these 
atoms that satisfy program rules according to ASP semantics. Our work 
focuses on the fundamental challenge of answer set counting \#ASP: 
determining the total number of valid answer sets for a given ASP 
program.

Answer set counting shares conceptual similarities with propositional
model counting (\#SAT), in which we count satisfying assignments of
Boolean formulas \citep{valiant1979}. While \#SAT is
\#\pclass-complete \citep{valiant1979}, its practical significance has
driven substantial research, yielding practically efficient
propositional model counters that combine strong theoretical
guarantees with impressive empirical performance. This, in turn, has
motivated research in counting techniques beyond propositional logic.
Specifically, there has been growing interest in answer set counting,
spurred by applications in probabilistic reasoning \citep{LTW2017},
network reliability analysis \citep{KM2023}, answer set
navigation~\citep{FGR2022,RHGGF2024}, and others~\citep{KTPM2024,KM2024,KM2025}.

Early approaches to answer set counting relied primarily on 
exhaustive enumeration \citep{GKS2012}. Recent methods have made 
significant progress by leveraging \#SAT techniques 
\citep{EHK2024,KCM2024,ACMS2015,JN2011,Janhunen2006,FGHR2024}. 
Complementing these approaches, dynamic programming on tree
decompositions has shown promise for programs with bounded treewidth
\citep{FHMW2017,FH2019}.  Most existing answer set counters focus on
normal logic programs --- a restricted class of ASP. Research on
counters for the more expressive class of disjunctive logic programs
\citep{EG1995} has received relatively less attention over the
years. Our work attempts to bridge this gap by focusing on practically
efficient counters for disjunctive logic programs. Complexity
theoretic arguments show that barring a collapse of the polynomial
hierarchy, translation from disjunctive to normal programs must incur
exponential overhead \citep{EFTW2004,Zhou2014}.  Consequently,
counters optimized for normal programs cannot efficiently handle
disjunctive programs, unless the programs themselves have {\em special
  properties} \citep{JWWWZX2016,FS2015,BAFP2017}. While loop
formula-based translation \citep{LL2003} enables counting in theory,
the exponential overhead becomes practically prohibitive for programs
with many cyclic atom relationships \citep{LR2006}.  Similarly,
although disjunctive answer set counting can be reduced to QBF
counting in principle~\citep{EETW2000}, this doesn't yield a
practically scalable counter since QBF model counting still does not
scale as well in practice as propositional model
counting~\citep{SMKS2022,CLPS2024}.  This leads to our central
research question: {\em Can we develop a practical answer set counter
  for disjunctive logic programs that scales effectively to handle
  large answer set counts? }

Our work provides an affirmative answer to this question through
several key contributions. We present the design, implementation, and
extensive evaluation of a novel counter for disjunctive programs,
employing subtractive reduction \citep{DHK2005} to projected
propositional model counting~\citep{ACMS2015a}, while maintaining
polynomial formula size growth. The approach first computes an
over-approximation of the answer set count, and then subtracts the
surplus computed using projected counting. This yields a \#\npclass algorithm
that leverages recent advances in projected propositional counting
\citep{SRSM2019,LM2019}. This approach is theoretically justified:
answer set counting for normal programs is in \#\pclass
\citep{JN2011,EHK2021}, while for disjunctive programs, it lies in
\#~$\cdot$~co-\npclass \citep{FHMW2017}. Since \#~$\cdot$~co-\npclass
= \#~$\cdot$~\pclass$^{\mathsf{NP}} =$ \#\npclass
\citep{DHK2005,HV1995}, our reduction is complexity-theoretically
sound and yields a practical counting algorithm.

While subtractive reduction for answer set counting has been 
proposed earlier \citep{HK2023}, our work makes several novel 
contributions beyond the theoretical framework. We develop a complete 
implementation with careful algorithm design choices and provide 
comprehensive empirical evaluation across diverse benchmarks. A 
detailed comparison with the prior approach is presented in Section 
\ref{section:counting}.

Our counter, \toolname, employs an alternative definition of answer 
sets for disjunctive programs, extending earlier work on normal programs 
\citep{KCM2024}. This definition enables the use of off-the-shelf 
projected model counters without exponential formula growth. Extensive 
experiments on standard benchmarks demonstrate that \toolname 
significantly outperforms existing counters on instances with large 
answer set counts. This motivates our development of a hybrid 
counter combining enumeration and \toolname to consistently exceed 
state-of-the-art performance.

The remainder of the paper is organized as follows. Section
\ref{section:preliminaries} covers essential background. Section
\ref{section:related_works} reviews prior work. Section
\ref{section:definition} presents our alternative answer set
definition for disjunctive programs. Section \ref{section:counting}
details our counting technique \toolname. Section
\ref{section:experiment} provides experimental results, and Section
\ref{section:conclusion} concludes the paper with future research
directions.

\section{Preliminaries}
\label{section:preliminaries}
We now introduce some notation and preliminaries needed in subsequent sections.

\paragraph{Propositional Satisfiability.}
A propositional \emph{variable} $v$ takes value from the domain $\{0, 1\}$ ($\{\false, \true\}$ resp.).  A
\emph{literal} $\ell$ is either a variable or its
negation.

A \emph{clause} $C$ is a {\em disjunction} ($\vee$) of literals.  For
clarity, we often represent a clause as a set of
literals, implicitly meaning that all literals in the
set are disjoined in the clause. 
A \emph{unit clause} is a clause with a single literal.
The constraint represented by a clause $C \equiv
(\neg{x_1} \vee \ldots \vee \neg{x_k} \vee x_{k+1} \vee \ldots \vee
x_{k+m})$ can be expressed as a logical \textit{implication} as follows: $(x_1
\wedge \ldots \wedge x_k) \longrightarrow (x_{k+1} \vee \ldots \vee
x_{k+m})$, where the conjuction of literals $x_1 \wedge \ldots \wedge x_k$ is known as the {\em antecedent} and the disjunction of literals is known as the {\em consequent}.
If $k = 0$, the antecedent of the implication is
$\true$, and if $m = 0$, the consequent is $\false$.

A formula $\phi$ is said to be in \emph{conjunctive normal form (CNF)} if it is a conjuction ($\wedge$)
of clauses. For convenience of exposition, a CNF formula is often represented as a set of clauses, 
implicitly meaning that all clauses in the set are conjoined in the
formula. We denote the set of variables of a propositional formula $\phi$ as $\var{\phi}$.

An assignment over a set $X$ of propositional variables is a mapping
$\tau: X \rightarrow \{0,1\}$.  For a variable $x \in X$, we define
$\tau(\neg{x}) = 1 - \tau(x)$. 
An assignment $\tau$ over $\var{\phi}$ is called a {\em model} of $\phi$, represented as $\tau \models \phi$, if $\phi$ evaluates to \true under the assignment $\tau$, as per the semantics of propositional logic.
A formula $\phi$ is said to be {\em SAT} (resp. {\em UNSAT}) if there exists a model (resp. no model) of $\phi$.
Given an assignment $\tau$, we use the notation $\tau^{+}$ (resp. $\tau^{-}$) to denote the set of variables that are assigned $1$ or \true (resp. $0$ or \false).

Given a CNF formula $\phi$ (as a set of clauses) and an assignment
$\tau: X \rightarrow \{0,1\}$, where $X \subseteq \var{\phi}$, the
\textit{unit propagation} of $\tau$ on $\phi$, denoted
$\up{\phi}{\tau}$, is another CNF formula obtained by applying the
following steps recursively: (a) remove each clause $C$ from $\phi$
that contains a literal $\ell$ s.t. $\tau(\ell) = 1$, (b) remove from
each clause $C$ in $\phi$ all literals $\ell$ s.t. either $\tau(\ell)
= 0$ or there exists a \emph{unit clause} $\{\neg \ell\}$, i.e. a
clause with a single literal $\neg \ell$, and (c) apply the above
steps recursively to the resulting CNF formula until there are no
further syntactic changes to the formula.
As a special case, the unit propagation of an empty
formula is the empty formula.
It is not hard to show that unit propagation of $\tau$ on $\phi$
always terminates or reaches {\em fixed point}. We say that $\tau$ \emph{unit propagates} to
literal $\ell$ in $\phi$, if $\{\ell\}$ is a unit clause in
$\up{\phi}{\tau}$, i.e. if $\{\ell\} \in \up{\phi}{\tau}$.

Given a propositional formula $\phi$, we use $\#\phi$ to denote the
count of models of $\phi$. If $X \subseteq \var{\phi}$ is a set of
variables, then $\#\exists X \phi$ denotes the count of models of
$\phi$ after disregarding assignments to the variables in
$X$. In other words, two different models of $\phi$
that differ only in the assignment of variables in $X$ are counted as
one in $\#\exists X \phi$.

\paragraph{Answer Set Programming.}
An \textit{answer set program} $P$ consists of a set of rules, where
each rule is structured as follows:
\begin{align}
\label{eq:general_rule}
\text{Rule $r$:~~}a_1 \vee \ldots a_k \leftarrow b_1, \ldots, b_m, \textsf{not } c_1, \ldots, \textsf{not } c_n
\end{align}
where $a_1, \ldots, a_k, b_1, \ldots, b_m, c_1, \ldots, c_n$ are
propositional variables or \emph{atoms}, and $k,m,n$ are non-negative
integers.  The notations $\rules{P}$ and $\at{P}$ refer to the rules
and atoms of the program $P$, respectively.  In rule $r$ above, the
operator ``\textsf{not}'' denotes \textit{default
  negation}~\citep{clark1978}. For each such rule $r$,
we use the following notation: the set of
atoms $\{a_1, \ldots, a_k\}$ constitutes the {\em head} of $r$,
denoted by $\head{r}$, the set of atoms $\{b_1, \ldots, b_m\}$ is
referred to as the {\em positive body atoms} of $r$, denoted by
$\body{r}^+$, and the set of atoms $\{c_1, \ldots, c_n\}$ is referred
to as the \textit{negative body atoms} of $r$, denoted by
$\body{r}^-$.
We use $\body{r}$ to denote the set of literals $\{b_1, \ldots, b_m,
\neg{c_1}, \ldots, \neg{c_n}\}$.  For notational convenience, we
sometimes use $\bot$ on the left (resp. $\top$ on the right) of $\leftarrow$ in a
rule $r$ to denote that $\head{r}$ (resp. $\body{r}$) is empty. A
program $P$ is called a {\em disjunctive logic program} if $\exists r
\in \rules{P}$ such that $\Card{\head{r}} \geq 2$~\citep{BD1994};
otherwise, it is a {\em normal logic program}.  Our
focus in this paper is on disjunctive logic programs.

Following standard ASP semantics, an interpretation $M$ over the atoms
$\at{P}$ specifies which atoms are present in $M$, or equivalently
assigned \true in $M$.  Specifically, atom $a$ is \true in $M$ if and
only if $a \in M$.
An interpretation $M$ satisfies a rule $r$, denoted by $M \models r$, if and only if $(\head{r} \cup \body{r}^{-}) \cap M \neq \emptyset$ or $\body{r}^{+} \setminus M \neq \emptyset$. 
An interpretation $M$ is a {\em model} (though not necessarily an answer set) of $P$, denoted by $M \models P$, if $M$ satisfies every rule in $P$, i.e., $\forall_{r \in \rules{P}} M \models r$. 
The \textit{Gelfond-Lifschitz (GL) reduct} of a program $P$ with respect to an interpretation $M$ is defined as $P^M = \{\head{r} \leftarrow \body{r}^+ \mid r \in \rules{P}, \body{r}^- \cap M = \emptyset\}$ \citep{GL1991}. 
An interpretation $M$ is an {\em answer set} of $P$ if $M \models P$ and $\not \exists M\textprime \subset M$ such that $M\textprime \models P^M$. 
In general, an ASP $P$ may have multiple answer sets.  The notation $\answer{P}$ denotes the set of all answer sets of $P$.

\paragraph{Clark Completion.}
The {\em Clark Completion}~\citep{LL2003} translates an ASP program $P$ to a propositional formula $\completion{P}$. 
The formula $\completion{P}$ is defined as the conjunction of the following propositional implications:
\begin{enumerate}
    \item \label{clause:g1} (group $1$) for each atom $a \in \at{P}$ s.t. $\not \exists r \in \rules{P}$ and $a \in \head{r}$, add a unit clause $\neg{a}$ to $\completion{P}$
    \item \label{clause:g2} (group $2$) for each rule $r \in \rules{P}$, add the following implication to $\completion{P}$:
    \[
        \bigwedge_{\ell \in \body{r}} \ell \longrightarrow \bigvee_{x \in \head{r}} x
    \]
    \item \label{clause:g3} (group $3$) for each atom $a \in \at{P}$ occuring in the head of at least one of the rules of $P$, let $r_1, \ldots, r_k$ be precisely all rules containing $a$ in the head, and add the following implication to $\completion{P}$:
    \[
        a \longrightarrow \bigvee_{i \in [1,k]} (\bigwedge_{\ell \in \body{r_i}} \ell ~\wedge~ \bigwedge_{x \in \head{r_i} \setminus \{a\}} \neg{x})
    \]
\end{enumerate}
It is known that every answer set of $P$ satisfies $\completion{P}$,
although the converse is not necessarily true~\citep{LL2003}.

Given a program $P$, we define the \textit{positive dependency graph} $\dependency{P}$ of $P$ as follows. 
Each atom $x \in \at{P}$ corresponds to a vertex in $\dependency{P}$. 
For $x, y \in \at{P}$, there is an edge from $y$ to $x$ in $\dependency{P}$ if there exists a rule $r \in \rules{P}$ such that $x \in \body{r}^+$ and $y \in \head{r}$ \citep{KS1992}. 
A set of atoms $L \subseteq \at{P}$ forms a \textit{loop} in $P$ if, for every $x, y \in L$, there is a path from $x$ to $y$ in $\dependency{P}$, and all atoms (equivalently, nodes) on the path belong to $L$. 
An atom $x$ is called a \textit{loop atom} of $P$ if there is a loop $L$ in $\dependency{P}$ such that $x \in L$. 
We use the notation $\loopatoms{P}$ to denote the set of all loop atoms of the program $P$. 
If there is no loop in $P$, we call the program \textit{tight}; otherwise, it is said to be \textit{non-tight}~\citep{Fages94}.
\begin{example}
    Consider the program $P = \{r_1: p_0 \vee p_1 \leftarrow \top; \text{ }r_2: q_0 \vee q_1 \leftarrow \top; \text{ }r_3: q_0 \leftarrow w; \text{ }r_4: q_1 \leftarrow w;\text{ }
    r_5: w \leftarrow p_0; \text{ }r_6: w \leftarrow p_1, q_1; \text{ }r_7: \bot \leftarrow \mathsf{not }\text{ }w;\}$.
    
    The group~\ref{clause:g2} clauses in $\completion{P}$ are: $\{(p_0 \vee p_1), (q_0 \vee q_1), (\neg{w} \vee q_0), (\neg{w} \vee q_1), (\neg{p_0} \vee w), (\neg{p_1} \vee \neg{q_1} \vee w), (w)\}$;
    and the group~\ref{clause:g3} clauses are: $\{(p_0 \longrightarrow \neg{p_1}), (p_1 \longrightarrow \neg{p_0}), (q_0 \longrightarrow (\neg{q_1} \vee w)), (q_1 \longrightarrow (\neg{q_0} \vee w)), (w \longrightarrow (p_0 \vee (p_1 \wedge q_1)))\}$.
    
    Since each atom occurs in at least one rule's head, there are no group~\ref{clause:g1} clauses. Thus, $\completion{P}$ consists of only group~\ref{clause:g2} and group~\ref{clause:g3} clauses. 
    In this program, the set of loop atoms is $\{q_1, w\}$.
\end{example}

\paragraph{Subtractive Reduction.} Borrowing notation from~\citep{DHK2005}, suppose $\Sigma$ and $\Gamma$ are alphabets, and $Q_1, Q_2 \subseteq \Sigma^* \times \Gamma^*$ are binary relations such that for each $x \in \Sigma^*$, the sets $Q_1(x) = \{y \in \Gamma^* \mid Q_1(x,y)\}$ and $Q_2(x) = \{y \in \Gamma^* \mid Q_2(x,y)\}$ are finite. Let $\#Q_1$ and $\#Q_2$ denote counting problems that require us to find $\Card{Q_1(x)}$ and $\Card{Q_2(x)}$ respectively, for a given $x \in \Sigma^*$.
We say that $\#Q_1$ strongly reduces to $\#Q_2$ via a subtractive
reduction, if there exist polynomial-time computable functions $f$ and
$g$ such that for every string $x \in \Sigma^*$, the following hold:
(a) $Q_2(g(x)) \subseteq Q_2(f(x))$, and
(b) $\Card{Q_1(x)} = \Card{Q_2(f(x))} - \Card{Q_2(g(x))}$.  As we will
see in Section~\ref{section:counting}, in our context, $\#Q_1$ is the
answer set counting problem for disjunctive logic programs, and
$\#Q_2$ is the projected model counting problem for propositional
formulas.

\section{Related Work}
\label{section:related_works}

Answer set counting exhibits distinct complexity characteristics across 
different classes of logic programs. For normal logic programs, the 
problem is {\#P}-complete~\citep{valiant1979}, while for disjunctive 
logic programs, it rises to {\#}~$\cdot$~{\coNP}~\citep{FHMW2017}. 
This complexity gap between normal and disjunctive programs highlights 
that answer set counting for disjunctive logic programs is likely 
harder than that for normal logic programs, under standard complexity 
theoretic assumptions.

This complexity distinction is also reflected in the corresponding decision 
problems as well. While determining the existence of an  
answer set for normal logic programs is {\npclass}-complete~\citep{MT1991}, 
the same problem for disjunctive logic programs is 
$\Sigma_{2}^{p}$-complete~\citep{EG1995}. This fundamental difference 
in complexity has important implications for translations between 
program classes. Specifically, a polynomial-time translation from 
disjunctive to normal logic programs that preserves the count of 
answer sets does exist unless the polynomial hierarchy 
collapses~\citep{JNSSY2006,Zhou2014,JWWWZX2016}.

Much of the early research on answer set counting focused on normal 
logic programs~\citep{EHK2021,EHK2024,KCM2024,ACMS2015}. The methodologies for 
counting answer sets have evolved significantly over time. Initial 
approaches relied primarily on enumerations~\citep{GKS2012}. More recent methods have adopted 
advanced algorithmic techniques, particularly tree decomposition and 
dynamic programming.
~\cite{FHMW2017} developed 
DynASP, an exact answer set counter optimized for instances with small treewidth. 
\cite{KESHFM2022} 
explored a different direction with ApproxASP, which implements an approximate counter providing  
$(\varepsilon,\delta)$-guarantees, with the adaptation of hashing-based 
techniques.

Subtraction-based techniques have emerged as 
promising approaches for various counting problems, e.g., MUS 
counting~\citep{BM2021}. In the 
context of answer set counting, subtraction-based methods were 
introduced in~\citep{HK2023,FGHR2024}. These methods employ a two-phase 
strategy: initially overcounts the answer set count, subsequently subtracts the surplus to obtain the exact count. 
\cite{HK2023} developed a method utilizing 
projected model counting over propositional formulas with projection 
sets. A detailed comparison of our work with their approach is 
provided at the end of Section~\ref{section:counting}. In a 
different direction, \cite{FGHR2024} proposed 
iascar, specifically tailored for normal programs. Their approach 
iteratively refines the overcount count by enforcing {\em external 
support} for each loop and applying the {\em inclusion-exclusion principle}. 
The key distinction of iascar lies in its comprehensive consideration 
of external supports for all cycles in the counting process.

\section{An Alternative Definition of Answer Sets}
\label{section:definition}
In this section, we present an alternative definition of answer sets
for disjunctive logic programs, that generalizes the work
of~\citep{KCM2024} for normal logic programs. Before presenting the alternative definition of answer sets, 
we provide a definition of {\em justification}, that is crucial to understand our technical contribution.

\subsection{Checking Justification in ASP}
Intuitively, justification refers to a {\em structured explanation} for {\em why} a literal (atom or its negation) is \true or \false in a given answer set~\citep{PSE2009,FS2019}.
Recall that the classical definition of answer sets requires that each
\true atom in an interpretation, that also appears at the head of a
rule, must be justified~\citep{GL1988,Lifschitz2010}. More precisely, given an
interpretation $M$ s.t. $M \models P$, ASP solvers check whether
some of the atoms in $M$ can be set to \false, while satisfying the
reduct program $P^{M}$~\citep{Lierler2005}.  We use the notation
$\tau_{M}$ to denote the assignment of propositional variables
corresponding to the interpretation $M$.  Furthermore, we say that $x \in
\tau_M^+$ (resp. $\tau_M^-$) iff $\tau_M(x) = 1$ (resp. $0$).

While the existing literature typically formulates justification using {\em rule-based} or {\em graph-based} explanations~\citep{FS2019}, we propose a model-theoretic definition from the reduct $P^M$, for each interpretation $M \models P$.
An atom $x \in M$ is {\em justified} in $M$ if for every $M\textprime \models P^M$ such that $M\textprime \subseteq M$, it holds that $x \in M\textprime$.
In other words, removing $x$ from $M$ violates the satisfaction of $P^M$.
The definition is compatible with the standard characterization of answer sets, since $M$ is an answer set, when no $M\textprime \subsetneq M$ exists such that $M\textprime \models P^M$; i.e., each atom $x \in M$ is justified.  
Conversely, an atom $x \in M$ is {\em not justified} in $M$ if there exists a proper subset $M\textprime \subset M$ such that $M\textprime \models P^M$ and $x \not\in M\textprime$.
This notion of justification also aligns with how SAT-based ASP solvers perform {\em minimality checks}~\citep{Lierler2005} --- such solvers encode $P^M$ as a set of implications (see definition of $P^M$ in Section~\ref{section:preliminaries})
 and check the satisfiability of the formula:
$P^M \wedge \bigwedge_{x \in \tau_M^{-}} \neg x \wedge \bigvee_{x \in \tau_M^{+}} \neg x$.

\begin{proposition}
    \label{prop:justification_all_atoms}
    For a program $P$ and each interpretation $M$ such that $M \models P$, if the formula $P^{M} \wedge \bigwedge_{x \in \tau_{M}^{-}} \neg{x} \wedge \bigvee_{x \in \tau_{M}^{+}} \neg{x}$ is satisfiable, then some atoms in $M$ are not justified.
\end{proposition}
The proposition holds by definition. In the above formula, the term $\bigwedge_{x \in \tau_{M}^{-}}
\neg{x}$ encodes the fact that variables assigned \false in $M$ need
no justification.  On the other hand, the term $\bigvee_{x \in
  \tau_{M}^{+}} \neg{x}$ verifies whether any of the variables
assigned \true in $M$ is not justified.

We now show that under the Clark completion of a program, or when
$\tau_{M} \models \completion{P}$, then it suffices to check
justification of only loop atoms of $P$ in the interpretation $M$.
Note that the ASP counter, sharpASP~\citep{KCM2024}, also checks
justifications for loop atoms in the context of normal logic programs.
Our contribution lies in proving the sufficiency of checking
justifications for loop atoms even in the context of disjunctive logic
programs -- a non-trivial generalization.  Specifically, we establish
that when $\tau_{M} \models \completion{P}$, if any atom in $M$ is not
justified, then there must also be some loop atoms in $M$ that is not
justified.  To verify justifications for only loop atoms, we check the
satisfiability of the formula: $P^{M} \wedge \bigwedge_{x \in
  \tau_{M}^{-}} \neg{x} \wedge \bigwedge_{x \in \tau_{M}^{+} \wedge x
  \not \in \loopatoms{P}} x \wedge \bigvee_{x \in \tau_{M}^{+} \wedge
  x \in \loopatoms{P}} \neg{x}$.

\begin{proposition}
    \label{prop:justification_loop_atoms}
    For each $M \subseteq \at{P}$ such that $M \models P$, if the formula 
    $P^{M} \wedge \bigwedge_{x \in \tau_{M}^{-}} \neg{x} \wedge \bigwedge_{x \in \tau_{M}^{+} \wedge x \not \in \loopatoms{P}} x \wedge \bigvee_{x \in \tau_{M}^{+} \wedge x \in \loopatoms{P}} \neg{x}$ 
    is satisfiable, then some of the loop atoms in $M$ are not justified
    .
\end{proposition}
\begin{proof}

  Since $\tau_M \models \completion{P}$, it implies that $\tau_M \models P^{M}$.
  Thus the formula $P^{M} \wedge \bigwedge_{x \in \tau_{M}^{-}} \neg{x} \wedge \bigwedge_{x \in \tau_{M}^{+} \wedge x \not \in \loopatoms{P}} x$ is satisfiable.
  
  If $P^{M} \wedge \bigwedge_{x \in \tau_{M}^{-}} \neg{x} \wedge \bigwedge_{x \in \tau_{M}^{+} \wedge x \not \in \loopatoms{P}} x \wedge \bigvee_{x \in \tau_{M}^{+} \wedge x \in \loopatoms{P}} \neg{x}$ is satisfiable, 
  then there are some loop atoms from $\tau_{M}^{+} \cup \loopatoms{P}$ that can be set to \false, while satisfying the formula $P^{M} \wedge \bigwedge_{x \in \tau_{M}^{-}} \neg{x} \bigwedge_{x \in \tau_{M}^{+} \wedge x \not \in \loopatoms{P}} x$. 
  It indicates that some of the loop atoms of $M$ are not justified; otherwise, each loop atom of $M$ is justified.
\end{proof}
In this above formula, the term $\bigwedge_{x \in \tau_{M}^{+} \wedge x \not
  \in \loopatoms{P}} x$ ensures that we are not concerned with
justifications for non-loop atoms.  On the other hand, the term
$\bigvee_{x \in \tau_{M}^{+} \wedge x \in \loopatoms{P}} \neg{x}$
specifically verifies whether any of the loop atoms assigned to \true in $M$ is not justified.

For every interpretation $M \models P$, checking justification of all loop atoms of $M$ suffices to check justification all atoms of $M$. 
The following lemma formalizes our claim: 
\begin{restatable}{lemma}{cyclicatomsuffices}
    \label{lemma:cyclic_atom_suffices}
    For a given program $P$ and each interpretation $M \subseteq \at{P}$ such that $\tau_{M} \models \completion{P}$,
      if $P^{M} \wedge \bigwedge_{x \in \tau_{M}^{-}} \neg{x} \wedge \bigvee_{x \in \tau_{M}^{+}} \neg{x}$ is SAT
      then $P^{M} \wedge \bigwedge_{x \in \tau_{M}^{-}} \neg{x} \wedge \bigwedge_{x \in \tau_{M}^{+} \wedge x \not \in \loopatoms{P}} x \wedge \bigvee_{x \in \tau_{M}^{+} \wedge x \in \loopatoms{P}} \neg{x}$ is also SAT. 
    \end{restatable}
\noindent \textit{The proof and illustrative examples are deferred to the Appendix.} 

\subsection{$\copyop{P}$ for Disjunctive Logic Programs}

Towards establishing an alternative definition of answer sets for
disjunctive logic programs, we now generalize the copy operation used
in~\citep{KCM2024} in the context of normal logic programs. Given an
ASP program $P$, for each loop atom $x \in \loopatoms{P}$, we
introduce a fresh variable $\copyatom{x}$ such that $\copyatom{x}
\not\in \at{P}$.  We refer to $\copyatom{x}$ as the \emph{copy
variable} of $x$.
Similar to~\citep{KCM2024,Kabir2024}, the operator $\copyop{P}$ returns the
following set of implicitly conjoined implications.
\begin{enumerate}
\item \label{l1:type1} (type $1$) for each loop atom $x \in \loopatoms{P}$, the implication $\copyatom{x} \longrightarrow x$ is included in $\copyop{P}$.
\item \label{l1:type2} (type $2$) for each rule $r = a_1 \vee \ldots a_k \leftarrow b_1, \ldots, b_m, \textsf{not } c_1, \ldots, \textsf{not } c_n \in \rules{P}$ such that 
  $\{a_1, \ldots a_k\} \cap \loopatoms{P} \neq \emptyset$, 
  the implication $\funcname{b_1} \wedge \ldots \funcname{b_m} \wedge \neg{c_1} \wedge \ldots \neg{c_n} \longrightarrow \funcname{a_1} \vee \ldots \funcname{a_k}$ is included in $\copyop{P}$, where $\funcname{x}$ is a function defined as follows:
  $\funcname{x} = \begin{cases}
    \copyatom{x} & \text{if $x \in \loopatoms{P}$}\\
    x & \text{otherwise}\\
\end{cases}$
\item No other implication is included in $\copyop{P}$.
\end{enumerate}
Note that we do not introduce any type~\ref{l1:type2} implication for a rule $r$ if $\head{r} \cap \loopatoms{P} = \emptyset$.
In a type~\ref{l1:type2} implication, each loop atom in the head and each positive body atom is replaced by its corresponding copy variable.
As a special case, if the program $P$ is tight then $\copyop{P} = \emptyset$.

We now demonstrate an important relationship between $P^{M}$ and
$\up{\copyop{P}}{\tau_M}$, for a given interpretation $M$.
Specifically, we show that we can use $\up{\copyop{P}}{\tau_M}$,
instead of $P^M$, to check the justification of loop atoms in $M$.  While sharpASP also utilizes a similar idea for normal programs, the following lemma
(\Cref{lemma:copy_reduct_are_equal}) formalizes this important
relationship in the context of the more general class of
disjunctive logic programs.

\begin{restatable}{lemma}{copyreductareequal}
    \label{lemma:copy_reduct_are_equal}
    For a given program $P$ and each interpretation $M \subseteq \at{P}$ such that $\tau_{M} \models \completion{P}$, 
    \begin{enumerate}
      \item the formula $\up{\copyop{P}}{\tau_{M}} \wedge \bigvee_{x \in \tau_{M}^{+} \wedge x \in \loopatoms{P}} \neg{\copyatom{x}}$ is SAT if and only if $P^{M} \wedge \bigwedge_{x \in \tau_{M}^{-}} \neg{x} \wedge \bigwedge_{x \in \tau_{M}^{+} \wedge x \not \in \loopatoms{P}} x \wedge \bigvee_{x \in \tau_{M}^{+} \wedge x \in \loopatoms{P}} \neg{x}$ is SAT 
      \item the formula $\up{\copyop{P}}{\tau_{M}} \wedge \bigvee_{x \in \tau_{M}^{+} \wedge x \in \loopatoms{P}} \neg{\copyatom{x}}$ is SAT if and only if $P^{M} \wedge \bigwedge_{x \in \tau_{M}^{-}} \neg{x} \wedge \bigvee_{x \in \tau_{M}^{+}} \neg{x}$ is SAT
    \end{enumerate}
\end{restatable}
\noindent \textit{The proof and illustrative examples are deferred to the Appendix.} 

We now integrate Clark's completion, the copy operation introduced above, and the core idea from \Cref{lemma:copy_reduct_are_equal} to propose an alternative definition of answer sets.
\begin{lemma}
    \label{prop:definition}
    For a given program $P$ and each interpretation $M \subseteq \at{P}$ such that $\tau_{M} \models \completion{P}$,
    $M \in \answer{P}$ if and only if the formula $\up{\copyop{P}}{\tau_{M}} \wedge \bigvee_{x \in \tau_{M}^{+} \wedge x \in \loopatoms{P}} \neg{\copyatom{x}}$ is UNSAT.
\end{lemma}
The proof follows directly from the correctness of Lemma~\ref{lemma:copy_reduct_are_equal}, and from the definition of answer sets based on the Gelfond-Lifschitz reduct $P^M$ (see~\Cref{section:preliminaries}).

Our alternative definition of answer sets, formalized
in~\Cref{prop:definition}, implies that the complexity of checking answer sets for disjunctive logic programs is in
\coNP.  In contrast, the definition in~\citep{KCM2024}, which applies
only to normal logic programs, allows answer set checking for this
restricted class of programs to be accomplished in polynomial time.
Note that the $\copyop{P}$ has similarities with formulas introduced in~\citep{FS2015,HK2023} for \coNP checks.

In the following section, we utilize the definition in
\Cref{prop:definition} to count of models of
$\completion{P}$ that are not answer sets of $P$. This approach allows
us to determine the number of answer sets of $P$ via subtractive
reduction.

\section{Answer Set Counting: \toolname}
\label{section:counting}
\begin{figure}
    \begin{center}
    \begin{tikzpicture}  
    \node[block, fill=orange] (a) {Compute $\first$};  
    \node[block,below=of a, fill=orange] (b) {Compute $\second$, $\project$};  
    \node[block,draw,fill=orange,align=center] at ([xshift=3cm,yshift=1.1cm]$(a)!1.0!(b)$)  (c) {\footnotesize $\#\first - \#\exists \project\second$};
    \draw[line] (a)-- (c);  
    \draw[line] (b)-- (c);
    \draw [line] ($(-2.5,-1.3cm)+(a)$)node[left]{$P$} -- (a);
    \draw [line] ($(-2.5,-1.3cm)+(a)$)node[left]{$P$} -- (b);
    \node[text width=3cm] at (0.6, 1.0) {Overcount};
    \node[text width=3cm] at (0.9,-3.2) {Surplus};
    \node[text width=3cm] at (3.6,-2.2) {Subtraction};
    \end{tikzpicture}  
    \end{center}
    \caption{The high-level architecture of \toolname~for a program $P$.}
    \label{figure:block_diagram}
\end{figure}
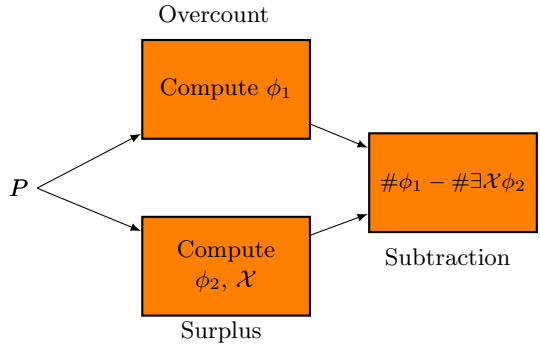

We now introduce a subtractive reduction-based technique for counting
the answer sets of disjunctive logic programs.  This approach reduces
answer set counting to projected model counting for propositional
formulas.  Note that projected model counting for propositional formulas is known to be in \#\npclass~\citep{ACMS2015a}; hence reducing
answer set counting (a \#~$\cdot$~co\npclass-complete problem)
to projected model counting makes sense\footnote{The classes
{\#\npclass} and {\#~$\cdot$~co\npclass} are known to coincide.}. In contrast,
answer set counting of normal logic programs is in \#P, and is
therefore easier.

At a high level, the proposed subtractive reduction
approach is illustrated in~\Cref{figure:block_diagram}.
For a given ASP program $P$, we overcount the answer sets of $P$
by considering the satisfying assignment of an appropriately
constructed propositional formula $\first$ (Overcount).  The
value $\#\first$ counts all answer sets of $P$, but also includes some
interpretations that are not answer sets of $P$.  To account for this
surplus, we introduce another Boolean formula $\second$
and a projection set $\project$ such that $\#\exists \project \second$
counts the surplus from the overcount of answer sets
(Surplus).
To correctly count the surplus, we employ the alternative answer set
definition outlined in~\Cref{prop:definition}.  Finally, the count of
answer sets of $P$ is determined by $\#\first - \#\exists\project\second$.

\paragraph{Counting Overcount}(\first)
Given a program $P$, the count of models of $\completion{P}$ provides
an overcount of the count of answer sets of $P$.  In the case
of tight programs, the count of answer sets is equivalent to the
count of models of $\completion{P}$~\citep{LL2003}.  However, for non
tight programs, the count of models of $\completion{P}$ 
overcounts $\Card{\answer{P}}$.  Therefore, we use 
\begin{align}
    \label{eq:exact_model_counting}
    \first = \completion{P}
\end{align}

\paragraph{Counting Surplus}(\second)
To count the surplus, we utilize the alternative answer set definition presented in~\Cref{prop:definition}. 
We use a propositional formula $\second$, in which for each loop atom $x$, there are two fresh copy variables: $\copyatom{x}$ and $\staratom{x}$. 
We introduce two sets of copy operations of $P$, namely, $\copyatom{\copyop{P}}$ and $\staratom{\copyop{P}}$, where for each loop atom $x$, the corresponding copy variables are denoted as $\copyatom{x}$ and $\staratom{x}$, respectively. 
We use the notations $\copyatom{\copyvar}$ and $\staratom{\copyvar}$ to refer to the copy variables of $\copyatom{\copyop{P}}$ and $\staratom{\copyop{P}}$, respectively;
i.e., $\copyatom{\copyvar} = \{\copyatom{x}|x \in \loopatoms{P}\}$ and $\staratom{\copyvar} = \{\staratom{x}|x \in \loopatoms{P}\}$.
To compute the surplus, we define the formula $\second(\at{P}, \copyatom{\copyvar}, \staratom{\copyvar})$ as follows:
\begin{align}
    \label{eq:projected_model_counting}
    \second(\at{P}, \copyatom{\copyvar}, &\staratom{\copyvar}) = \completion{P} \wedge \copyatom{\copyop{P}} \wedge \staratom{\copyop{P}} \nonumber \\&\wedge \bigwedge_{x \in \loopatoms{P}} (\copyatom{x} \longrightarrow \staratom{x}) \wedge \bigvee_{x \in \loopatoms{P}} (\neg{\copyatom{x}} \wedge \staratom{x})
\end{align}

\begin{restatable}{lemma}{counting}
  \label{lemma:counting}
  The number of models of $\completion{P}$ that are not answer sets of $P$
  can be computed as $\#\exists \copyatom{\copyvar}, \staratom{\copyvar}\text{ }\second(\at{P}, \copyatom{\copyvar}, \staratom{\copyvar})$, where the formula \second~is defined in~\Cref{eq:projected_model_counting}.
\end{restatable}
\begin{proof}
  From the definition of {\second} (\Cref{eq:projected_model_counting}), we know that for every model $\sigma \models \second$, the assignment to $\copyatom{\copyvar}$ and $\staratom{\copyvar}$ is such that  
  $\forall x \in \loopatoms{P}, \sigma(\copyatom{x}) \leq \sigma(\staratom{x})$ and $\exists x \in \loopatoms{P}, \sigma(\copyatom{x}) < \sigma(\staratom{x})$\footnote{We use $0 < 1$ for this discussion.}.
  Let $M$ be the corresponding interpretation over $\at{P}$ of the satisfying assignment $\sigma$.
  Since $\sigma \models \second$, $\tau_M \models \completion{P}$ and some of the copy variables $\copyatom{x} \in \copyatom{\copyvar}$ can be set to \false where $\sigma(x) = \true$, while after setting the copy variables $\copyatom{x}$ to \false, the formula $\copyop{P}_{|\tau_M}$ is still satisfied.
  According to~\Cref{prop:definition}, we can conclude that $M \not \in \answer{P}$. 
  As a result, $\#\exists \copyatom{\copyvar}, \staratom{\copyvar}\text{ }\second(\at{P}, \copyatom{\copyvar}, \staratom{\copyvar})$ counts all interpretations that are not answer sets of $P$.
\end{proof}

\begin{restatable}{theorem}{maintheorem}
  \label{theorem:counting}
    For a given program $P$, the number of answer sets: $\Card{\answer{P}} = \#\first - \#\exists\project\second$, where $\project = \copyatom{CV} \cup \staratom{CV}$, and \first~and \second~are defined in~\Cref{eq:exact_model_counting,eq:projected_model_counting}. Furthermore, both \first~ and \second~ can be computed in polynomial time in $\Card{P}$.
\end{restatable}
\begin{proof}
  The proof of the part $\Card{\answer{P}} = \#\first - \#\exists\project\second$, where $\project = \copyatom{CV} \cup \staratom{CV}$, follows from~\Cref{lemma:counting}.
  Initially, $\#\first$ overcounts the number of answer sets, while the~\Cref{lemma:counting} establishes that $\#\exists\project\second$ counts the surplus from $\#\first$. Thus, the subtraction determines the number of answer sets of $P$. 
  
  For a program $P$, $\completion{P}$, $\copyop{P}$, and $\loopatoms{P}$ can be computed in time polynomial in the size of $P$. 
  Additionally, for a program $P$, $\Card{\loopatoms{P}} \leq \Card{\at{P}}$. 
  Thus, we can compute both $\first$ and $\second$ in polynomial time in the size of $P$.
\end{proof}

Now recall to subtractive reduction definition (ref.~\Cref{section:preliminaries}),
for a given ASP program $P$, $f(P)$ computes the
formula $\first$, and $g(P)$ computes the formula $\exists
\mathsf{CV}',\mathsf{CV}^*\second$.

We refer to the answer set counting technique based
on~\Cref{theorem:counting} as \toolname.  While \toolname shares
similarities with the answer set counting approach
outlined in~\citep{HK2023}, there are key differences between the two
techniques.  First, instead of counting the number of models of Clark
completion, the technique in~\citep{HK2023} counts {\em
  non-models} of the Clark completion.  Second, to count the surplus,
\toolname~introduces copy variables only for loop variables, whereas
the approach of~\citep{HK2023} introduces copy (referred to as {\em duplicate} variable) variables for
every variable in the program.  Third, \toolname~focuses on generating
a copy program over the cyclic components of the input program, while
their approach duplicates the entire program.  
A key distinction is that the size of Boolean formulas introduced by~\citep{HK2023} depends on the tree decomposition of the input
program and its treewidth, assuming that the treewidth is small.
However, most natural encodings that result in ASP programs are not
treewidth-aware~\citep{Hecher2022}. 
Importantly, their work focused on theoretical treatment and, as such, does not address algorithmic aspects. It is worth noting that there is no accompanying implementation. Our personal communication with authors confirmed that they have not yet implemented their proposed technique.

\section{Experimental Results}
\label{section:experiment}

We developed a prototype of \toolname\footnote{\url{https://github.com/meelgroup/SharpASP-SR}}, by leveraging existing 
projected model counters. Specifically, we employed \ganak~\citep{SRSM2019} 
as the underlying projected model counter, given its competitive 
performance in model counting competitions. The evaluation with 
All counters are sourced from the model counting competition $2024$.

\paragraph{Baseline and Benchmarks.}
We evaluated \toolname~against state-of-the-art ASP systems capable of 
handling disjunctive answer set programs: (i)~\clingo v$5.7.1$~\citep{GKS2012}, 
(ii)~DynASP v$2.0$~\citep{FHMW2017}, and (iii)~Wasp v$2$~\citep{ADLR2015}. ASP 
solvers \clingo~and Wasp count answer sets via enumeration. We were 
unable to baseline against existing ASP counters such as 
aspmc+\#SAT~\citep{EHK2024}, lp2sat+\#SAT~\citep{Janhunen2006,JN2011}, 
sharpASP~\citep{KCM2024}, and iascar~\citep{FGHR2024}, as these systems 
are designed exclusively for counting answer sets of normal logic 
programs. Since no implementation is available for the counting 
techniques outlined by~\citep{HK2023}, a comparison against their approach 
was not possible. We also considered ApproxASP~\citep{KESHFM2022} for 
comparison purposes, with results deferred to the appendix.

Our benchmark suite comprised non-tight disjunctive logic program 
instances previously used to evaluate disjunctive answer set solvers. 
These benchmarks span diverse computational problems, including: 
(i)~$2$QBF~\citep{KESHFM2022}, (ii)~strategic companies~\citep{Lierler2005}, 
(iii)~{\em preferred} extensions of {\em abstract argumentation}~\citep{GMRWW2015}, 
(iv)~pc configuration~\citep{FGR2022}, (v)~minimal 
diagnosis~\citep{GSTUV2008}, and (vi)~{\em minimal trap spaces}~\citep{TBPS2024}. The benchmarks were sourced from abstract argumentation 
competitions, ASP competitions~\citep{GMR2020} and from~\citep{KESHFM2022,TBPS2024}. 
Following recent work on disjunctive logic programs~\citep{AADLMR2019}, 
we generated additional non-tight disjunctive answer set programs using 
the generator implemented by~\citep{ART2017}. 
The complete benchmark set 
comprises $1125$ instances, available in supplementary materials.

\paragraph{Environmental Settings.}
All experiments were conducted on a computing cluster equipped with AMD 
EPYC $7713$ processors. Each benchmark instance was allocated one core, 
with runtime and memory limits set to $5000$ seconds and $8$ GB 
respectively for all tools, which is consistent with prior works on model counting and answer set counting.

\subsection{Experimental Results}

\begin{table}[h]
    \centering
    \begin{tabular}{m{7em} m{3em} m{4em} m{3em} m{6em}} 
    \toprule
    & \clingo & DynASP & Wasp & \toolname\\
    \midrule
    \#Solved {\small ($1125$)} & 708 & 89 & 432 & \textbf{825}\\
    \midrule
    PAR$2$ & 4118 & 9212 & 6204 & \textbf{2939}\\
    \bottomrule
    \end{tabular}
    \caption{The performance of \toolname~vis-a-vis existing disjunctive answer set counters, based on $1125$ instances.}
    \label{table:experimental_result}
    \vspace{-0.8em}
\end{table}

\begin{table}[h]
  \centering
  \begin{tabular}{m{7em} m{3em} m{4em} m{3em} m{6em}} 
  \toprule
  & & \multicolumn{3}{|c|}{\clingo ($\leq 10^4$) + }\\
  & \clingo & DynASP & Wasp & \toolname\\
  \midrule
  \#Solved {\small ($1125$)} & 708 & 377 & 442 & \textbf{918}\\
  \midrule
  PAR$2$ & 4118 & 4790 & 4404 & \textbf{1600}\\
  \bottomrule
  \end{tabular}
  \caption{The performance comparison of hybrid counters, based on $1125$ instances.
  The hybrid counters correspond to last $3$ columns that employ clingo enumeration followed by ASP counters.
  The clingo ($2$nd column) refers to clingo enumeration for $5000$ seconds.}
  \label{table:hybrid_counter_result}
  \vspace{-0.8em}
\end{table}

\toolname~demonstrated significant performance improvement across the benchmark suite, 
as evidenced in Table~\ref{table:experimental_result}. For comparative 
analysis, we present both the number of solved instances and PAR$2$ 
scores~\citep{SAT2017}, for each tool. \toolname~achieved the highest solution count 
while maintaining the lowest PAR$2$ score, indicating superior 
scalability compared to existing systems capable of counting answer 
sets of disjunctive logic programs. The comparative performance of different counters is shown in a cactus plot in~\Cref{fig:runtime_counters}.

Given \clingo's superior performance on instances with few answer sets, 
we developed a hybrid counter integrating the strengths of \clingo's enumeration 
and other counting techniques, following the experimental evaluation of~\citep{KCM2024}. This hybrid approach 
first employs \clingo~enumeration (maximum $10^4$ answer sets) and 
switches to alternative counting techniques if needed. 
Within our benchmark instances, a noticeable shift was observed on \clingo's runtime performance when the number of answer sets exceeds $10^4$. 
As shown in 
Table~\ref{table:hybrid_counter_result}, the hybrid counter based on 
\toolname~significantly outperforms baseline approaches.

\begin{figure}
  \centering
  \scalebox{0.8}{
    \includegraphics[width=0.9\linewidth]{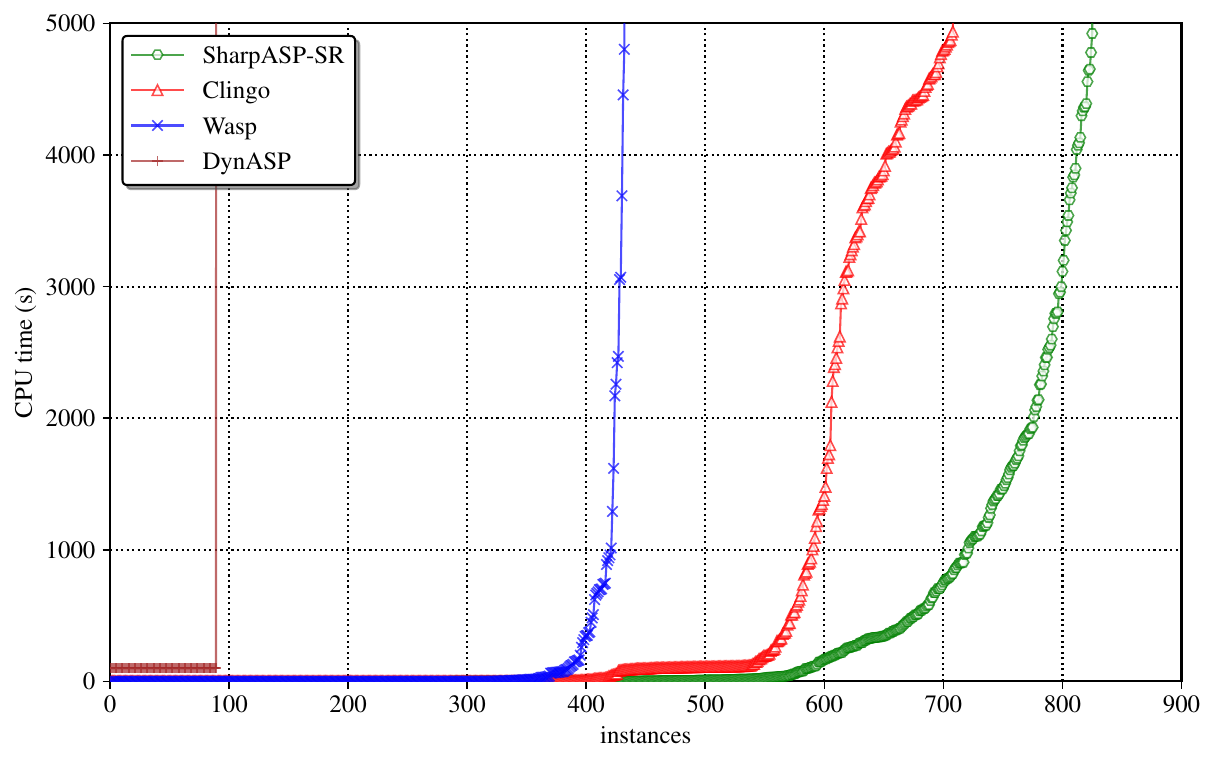}
  }
    \caption{The runtime performance of \toolname~vis-a-vis other ASP counters.}
    \label{fig:runtime_counters}
  \vspace{-0.8em}
\end{figure}
The cactus plot in Figure~\ref{fig:runtime_counters} illustrates the 
runtime performance of the four tools, where a point $(x,y)$ indicates 
that a tool can count $x$ instances within $y$ seconds. The plot 
shows \toolname's clear performance advantage over state-of-the-art 
answer set counters for disjunctive logic programs.

\begin{table}[h]
  \centering
  \begin{tabular}{m{4em} m{2em} m{4em} m{4em} m{4em} m{6em} m{10em}} 
  \toprule
  $\Card{\loopatoms{P}}$ & $\sum$ & \clingo & DynASP & Wasp & \toolname & \clingo+\toolname\\
  \midrule
  $[1,100]$ & 399 & 248 & 87 & 165 & \textbf{386} & 388\\
  \midrule
  $[101,1000]$ & 519 & 316 & 2 & 142 & \textbf{398} & 401\\
  \midrule
  $> 1000$ & 207 & \textbf{144} & 0 & 125 & 41 & 129\\
  \bottomrule
  \end{tabular}
  \caption{The performance comparison of \toolname~(SA) vis-a-vis existing disjunctive answer set counters across instances with varying numbers of loop atoms. 
  The second column ($\sum$) indicates the number of instances within each range of $\Card{\loopatoms{P}}$.}
  \label{table:higher_cyclic_result}
  \vspace{-0.8em}
\end{table}

Since \clingo~and Wasp employ enumeration-based techniques, their 
performance is inherently constrained by the answer set count. Our analysis 
revealed that \clingo~(resp. Wasp) timed out on nearly all instances 
with approximately $2^{30}$ (resp. $2^{24}$) or more answer sets, while \toolname~can count instances upto $2^{127}$ answer sets.
However, the performance of \toolname~is primarily influenced by the hardness of the projected model counting, which is related to the cyclicity of the program. 
The cyclicity of a program is quantified using the measure $\Card{\loopatoms{P}}$.

To analyze \toolname's performance relative to $\Card{\loopatoms{P}}$, 
we compared different ASP counters across varying ranges of loop atoms. 
The results in the Table~\ref{table:higher_cyclic_result} indicate that while 
\toolname~performs exceptionally well on instances with fewer loop 
atoms, its performance deteriorates significantly for instances with a higher number of loop atoms (e.g., those with $\Card{\loopatoms{P}} >1000$),  
leading to a decrease in the solved instances count.

\begin{figure}
  \centering
          \includegraphics[width=0.6\linewidth]{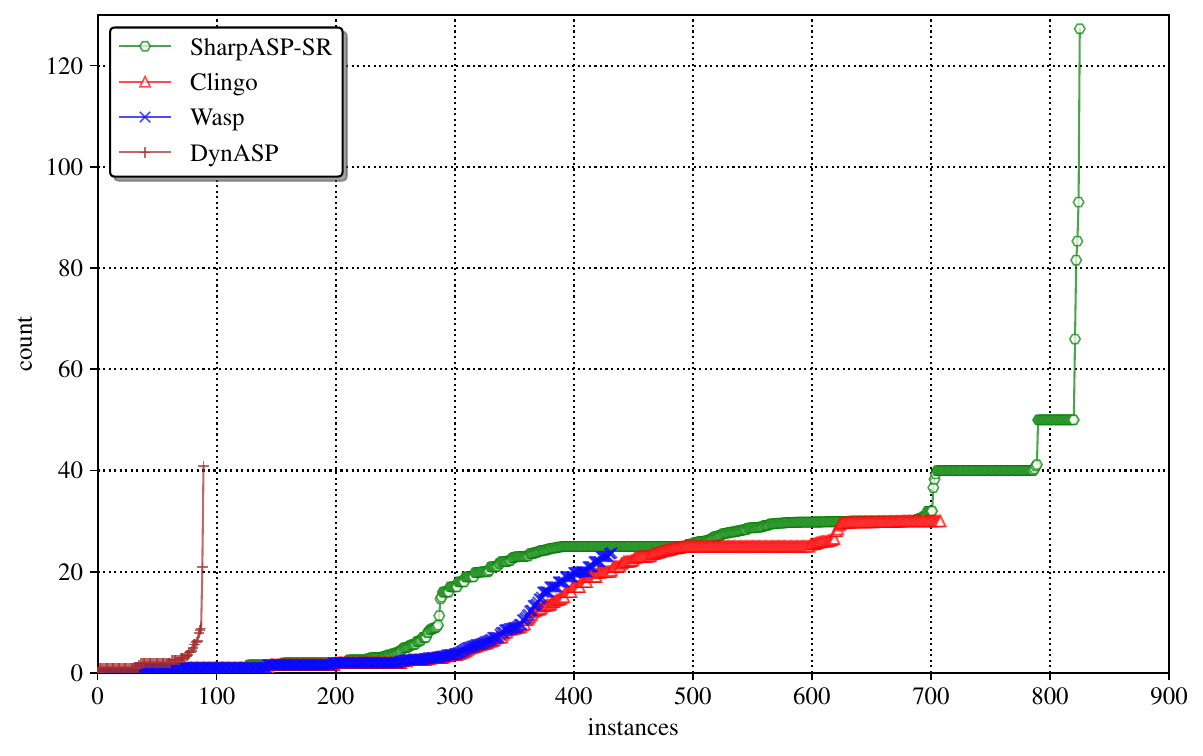}
      \caption{The performance comparison of \toolname~in terms of the number of answer sets. The $x$ and $y$-axis indicate the number of instances and answer set count (in $2$-base log scale), respectively.}
      \label{fig:count_counters}
\end{figure}
To further analyze the performance of \toolname, we compare the number of answer sets for each instance solved by different ASP counters.
We visually present these comparisons in~\Cref{fig:count_counters}.
In the plot, a point $(x,y)$ indicates that a tool can count $x$ instances, where each instance has up to $2^y$ answer sets.  
Since both \clingo and Wasp count using enumeration, \clingo and Wasp ~can handle instances with up to $2^{30}$ and $2^{24}$ answer sets (roughly) respectively, 
whereas \toolname~is capable of counting instances having $2^{127}$ answer sets. 
Due to its use of projected model counting, \toolname~demonstrates superior scalability on instances with a large number of answer sets.

\noindent {\em Further experimental evaluation of \toolname~is provided in the~\ref{section:detailed_experimental_analysis}.}

\section{Conclusion}
\label{section:conclusion}

In this paper, we introduced \toolname, a novel answer set counter based on subtractive reduction. By leveraging an alternative definition of answer sets, \toolname achieves significant performance improvements over baseline approaches, owing to its ability to rely on the scalability of state-of-the-art projected model counting techniques. Our experimental results demonstrate the effectiveness and efficiency of our approach across a range of benchmarks.

The use of subtractive reductions for empirical efficiency opens up potential avenues for future work. In particular, an interesting direction would be to categorize problems that can be reduced to \#SAT via subtractive methods, which would allow us to utilize existing \#SAT model counters.

\bibliography{tlpguide}

\begin{thebibliography}{}

\bibitem[Alviano et~al., 2019]{AADLMR2019}
{\sc Alviano, M.}, {\sc Amendola, G.}, {\sc Dodaro, C.}, {\sc Leone, N.}, {\sc
  Maratea, M.}, {\sc and} {\sc Ricca, F.}
\newblock Evaluation of disjunctive programs in {W}asp.
\newblock In {\em LPNMR} 2019, pp. 241--255. Springer.

\bibitem[Alviano et~al., 2015]{ADLR2015}
{\sc Alviano, M.}, {\sc Dodaro, C.}, {\sc Leone, N.}, {\sc and} {\sc Ricca, F.}
\newblock Advances in {W}asp.
\newblock In {\em LPNMR} 2015, pp. 40--54. Springer.

\bibitem[Amendola et~al., 2017]{ART2017}
{\sc Amendola, G.}, {\sc Ricca, F.}, {\sc and} {\sc Truszczynski, M.}
\newblock Generating hard random boolean formulas and disjunctive logic
  programs.
\newblock In {\em IJCAI} 2017, pp. 532--538.

\bibitem[Aziz et~al., 2015a]{ACMS2015a}
{\sc Aziz, R.~A.}, {\sc Chu, G.}, {\sc Muise, C.}, {\sc and} {\sc Stuckey, P.}
\newblock $\#\exists ${SAT}: Projected model counting.
\newblock In {\em SAT} 2015a, pp. 121--137. Springer.

\bibitem[Aziz et~al., 2015b]{ACMS2015}
{\sc Aziz, R.~A.}, {\sc Chu, G.}, {\sc Muise, C.}, {\sc and} {\sc Stuckey,
  P.~J.}
\newblock Stable model counting and its application in probabilistic logic
  programming.
\newblock In {\em AAAI} 2015b.

\bibitem[Balyo et~al., 2017]{SAT2017}
{\sc Balyo, T.}, {\sc Heule, M.~J.}, {\sc and} {\sc J{\"a}rvisalo, M.}
\newblock {SAT} competition 2017--solver and benchmark descriptions.
\newblock pp. 14--15 2017.

\bibitem[Ben-Eliyahu and Dechter, 1994]{BD1994}
{\sc Ben-Eliyahu, R.} {\sc and} {\sc Dechter, R.}
\newblock Propositional semantics for disjunctive logic programs.
\newblock {\em Annals of Mathematics and Artificial intelligence}, 12:53--87
  1994.

\bibitem[Ben-Eliyahu-Zohary et~al., 2017]{BAFP2017}
{\sc Ben-Eliyahu-Zohary, R.}, {\sc Angiulli, F.}, {\sc Fassetti, F.}, {\sc and}
  {\sc Palopoli, L.}
\newblock Modular construction of minimal models.
\newblock In {\em LPNMR} 2017, pp. 43--48. Springer.

\bibitem[Bend{\'\i}k and Meel, 2021]{BM2021}
{\sc Bend{\'\i}k, J.} {\sc and} {\sc Meel, K.~S.}
\newblock Counting minimal unsatisfiable subsets.
\newblock In {\em CAV} 2021, pp. 313--336. Springer.

\bibitem[Capelli et~al., 2024]{CLPS2024}
{\sc Capelli, F.}, {\sc Lagniez, J.-M.}, {\sc Plank, A.}, {\sc and} {\sc Seidl,
  M.}
\newblock A top-down tree model counter for quantified boolean formulas.
\newblock {\em IJCAI} 2024.

\bibitem[Clark, 1978]{clark1978}
{\sc Clark, K.~L.}
\newblock Negation as failure.
\newblock {\em Logic and data bases}, pp. 293--322 1978.

\bibitem[Durand et~al., 2005]{DHK2005}
{\sc Durand, A.}, {\sc Hermann, M.}, {\sc and} {\sc Kolaitis, P.~G.}
\newblock Subtractive reductions and complete problems for counting complexity
  classes.
\newblock {\em Theoretical Computer Science}, 340(3):496--513 2005.

\bibitem[Egly et~al., 2000]{EETW2000}
{\sc Egly, U.}, {\sc Eiter, T.}, {\sc Tompits, H.}, {\sc and} {\sc Woltran, S.}
\newblock Solving advanced reasoning tasks using quantified boolean formulas.
\newblock In {\em AAAI/IAAI} 2000, pp. 417--422.

\bibitem[Eiter et~al., 2004]{EFTW2004}
{\sc Eiter, T.}, {\sc Fink, M.}, {\sc Tompits, H.}, {\sc and} {\sc Woltran, S.}
\newblock On eliminating disjunctions in stable logic programming.
\newblock {\em KR}, 4:447--458 2004.

\bibitem[Eiter and Gottlob, 1995]{EG1995}
{\sc Eiter, T.} {\sc and} {\sc Gottlob, G.}
\newblock On the computational cost of disjunctive logic programming:
  Propositional case.
\newblock {\em Annals of Mathematics and Artificial Intelligence}, 15:289--323
  1995.

\bibitem[Eiter et~al., 2021]{EHK2021}
{\sc Eiter, T.}, {\sc Hecher, M.}, {\sc and} {\sc Kiesel, R.}
\newblock Treewidth-aware cycle breaking for algebraic answer set counting.
\newblock In {\em KR} 2021, pp. 269--279.

\bibitem[Eiter et~al., 2024]{EHK2024}
{\sc Eiter, T.}, {\sc Hecher, M.}, {\sc and} {\sc Kiesel, R.}
\newblock aspmc: New frontiers of algebraic answer set counting.
\newblock {\em Artificial Intelligence}, 330:104109 2024.

\bibitem[Fages, 1994]{Fages94}
{\sc Fages, F.}
\newblock Consistency of {C}lark's completion and existence of stable models.
\newblock {\em Journal of Methods of logic in computer science}, 1(1):51--60
  1994.

\bibitem[Fandinno and Schulz, 2019]{FS2019}
{\sc Fandinno, J.} {\sc and} {\sc Schulz, C.}
\newblock Answering the “why” in answer set programming--a survey of
  explanation approaches.
\newblock {\em TPLP}, 19(2):114--203 2019.

\bibitem[Fichte et~al., 2024]{FGHR2024}
{\sc Fichte, J.~K.}, {\sc Gaggl, S.~A.}, {\sc Hecher, M.}, {\sc and} {\sc
  Rusovac, D.}
\newblock {IASCAR}: Incremental answer set counting by anytime refinement.
\newblock {\em TPLP}, 24(3):505--532 2024.

\bibitem[Fichte et~al., 2022]{FGR2022}
{\sc Fichte, J.~K.}, {\sc Gaggl, S.~A.}, {\sc and} {\sc Rusovac, D.}
\newblock Rushing and strolling among answer sets--navigation made easy.
\newblock In {\em AAAI} 2022, volume~36, pp. 5651--5659.

\bibitem[Fichte and Hecher, 2019]{FH2019}
{\sc Fichte, J.~K.} {\sc and} {\sc Hecher, M.}
\newblock Treewidth and counting projected answer sets.
\newblock In {\em LPNMR} 2019, pp. 105--119. Springer.

\bibitem[Fichte et~al., 2017]{FHMW2017}
{\sc Fichte, J.~K.}, {\sc Hecher, M.}, {\sc Morak, M.}, {\sc and} {\sc Woltran,
  S.}
\newblock Answer set solving with bounded treewidth revisited.
\newblock In {\em {LPNMR}} 2017, pp. 132--145.

\bibitem[Fichte and Szeider, 2015]{FS2015}
{\sc Fichte, J.~K.} {\sc and} {\sc Szeider, S.}
\newblock Backdoors to normality for disjunctive logic programs.
\newblock {\em TOCL}, 17(1):1--23 2015.

\bibitem[Gaggl et~al., 2015]{GMRWW2015}
{\sc Gaggl, S.~A.}, {\sc Manthey, N.}, {\sc Ronca, A.}, {\sc Wallner, J.~P.},
  {\sc and} {\sc Woltran, S.}
\newblock Improved answer set programming encodings for abstract argumentation.
\newblock {\em TPLP}, 15(4-5):434--448 2015.

\bibitem[Gebser et~al., 2012]{GKS2012}
{\sc Gebser, M.}, {\sc Kaufmann, B.}, {\sc and} {\sc Schaub, T.}
\newblock Conflict-driven answer set solving: From theory to practice.
\newblock {\em Artificial Intelligence}, 187:52--89 2012.

\bibitem[Gebser et~al., 2020]{GMR2020}
{\sc Gebser, M.}, {\sc Maratea, M.}, {\sc and} {\sc Ricca, F.}
\newblock The seventh answer set programming competition: Design and results.
\newblock {\em TPLP}, 20(2):176--204 2020.

\bibitem[Gebser et~al., 2008]{GSTUV2008}
{\sc Gebser, M.}, {\sc Schaub, T.}, {\sc Thiele, S.}, {\sc Usadel, B.}, {\sc
  and} {\sc Veber, P.}
\newblock Detecting inconsistencies in large biological networks with answer
  set programming.
\newblock In {\em ICLP} 2008, pp. 130--144. Springer.

\bibitem[Gelfond and Lifschitz, 1988]{GL1988}
{\sc Gelfond, M.} {\sc and} {\sc Lifschitz, V.}
\newblock The stable model semantics for logic programming.
\newblock In {\em ICLP/SLP} 1988, volume~88, pp. 1070--1080.

\bibitem[Gelfond and Lifschitz, 1991]{GL1991}
{\sc Gelfond, M.} {\sc and} {\sc Lifschitz, V.}
\newblock Classical negation in logic programs and disjunctive databases.
\newblock {\em New generation computing}, 9:365--385 1991.

\bibitem[Hecher, 2022]{Hecher2022}
{\sc Hecher, M.}
\newblock Treewidth-aware reductions of normal {ASP} to {SAT}--is normal {ASP}
  harder than {SAT} after all?
\newblock {\em Artificial Intelligence}, 304:103651 2022.

\bibitem[Hecher and Kiesel, 2023]{HK2023}
{\sc Hecher, M.} {\sc and} {\sc Kiesel, R.}
\newblock The impact of structure in answer set counting: fighting cycles and
  its limits.
\newblock In {\em KR} 2023, pp. 344--354.

\bibitem[Hemaspaandra and Vollmer, 1995]{HV1995}
{\sc Hemaspaandra, L.~A.} {\sc and} {\sc Vollmer, H.}
\newblock The satanic notations: counting classes beyond $\#${P} and other
  definitional adventures.
\newblock {\em ACM SIGACT News}, 26(1):2--13 1995.

\bibitem[Janhunen, 2006]{Janhunen2006}
{\sc Janhunen, T.}
\newblock Some (in) translatability results for normal logic programs and
  propositional theories.
\newblock {\em Journal of Applied Non-Classical Logics}, 16(1-2):35--86 2006.

\bibitem[Janhunen and Niemel{\"a}, 2011]{JN2011}
{\sc Janhunen, T.} {\sc and} {\sc Niemel{\"a}, I.} 2011.
\newblock {\em Compact Translations of Non-disjunctive Answer Set Programs to
  Propositional Clauses}, pp. 111--130.

\bibitem[Janhunen et~al., 2006]{JNSSY2006}
{\sc Janhunen, T.}, {\sc Niemel{\"a}, I.}, {\sc Seipel, D.}, {\sc Simons, P.},
  {\sc and} {\sc You, J.-H.}
\newblock Unfolding partiality and disjunctions in stable model semantics.
\newblock {\em TOCL}, 7(1):1--37 2006.

\bibitem[Ji et~al., 2016]{JWWWZX2016}
{\sc Ji, J.}, {\sc Wan, H.}, {\sc Wang, K.}, {\sc Wang, Z.}, {\sc Zhang, C.},
  {\sc and} {\sc Xu, J.}
\newblock Eliminating disjunctions in answer set programming by restricted
  unfolding.
\newblock In {\em IJCAI} 2016, pp. 1130--1137.

\bibitem[Kabir, 2024]{Kabir2024}
{\sc Kabir, M.}
\newblock Minimal model counting via knowledge compilation.
\newblock {\em arXiv preprint arXiv:2409.10170} 2024.

\bibitem[Kabir et~al., 2024]{KCM2024}
{\sc Kabir, M.}, {\sc Chakraborty, S.}, {\sc and} {\sc Meel, K.~S.}
\newblock Exact {ASP} counting with compact encodings.
\newblock In {\em AAAI} 2024, volume~38, pp. 10571--10580.

\bibitem[Kabir et~al., 2022]{KESHFM2022}
{\sc Kabir, M.}, {\sc Everardo, F.~O.}, {\sc Shukla, A.~K.}, {\sc Hecher, M.},
  {\sc Fichte, J.~K.}, {\sc and} {\sc Meel, K.~S.}
\newblock {ApproxASP}--a scalable approximate answer set counter.
\newblock In {\em AAAI} 2022, pp. 5755--5764.

\bibitem[Kabir and Meel, 2023]{KM2023}
{\sc Kabir, M.} {\sc and} {\sc Meel, K.~S.}
\newblock A fast and accurate {ASP} counting based network reliability
  estimator.
\newblock In {\em LPAR} 2023, pp. 270--287.

\bibitem[Kabir and Meel, 2024]{KM2024}
{\sc Kabir, M.} {\sc and} {\sc Meel, K.~S.}
\newblock On lower bounding minimal model count.
\newblock {\em TPLP}, 24(4):586--605 2024.

\bibitem[Kabir and Meel, 2025]{KM2025}
{\sc Kabir, M.} {\sc and} {\sc Meel, K.~S.}
\newblock An {ASP}-based framework for {MUS}es.
\newblock {\em arXiv preprint arXiv:2507.03929} 2025.

\bibitem[Kabir et~al., 2025]{KTPM2024}
{\sc Kabir, M.}, {\sc Trinh, V.-G.}, {\sc Pastva, S.}, {\sc and} {\sc Meel,
  K.~S.}
\newblock Scalable counting of minimal trap spaces and fixed points in boolean
  networks.
\newblock {\em arXiv preprint arXiv:2506.06013} 2025.

\bibitem[Kanchanasut and Stuckey, 1992]{KS1992}
{\sc Kanchanasut, K.} {\sc and} {\sc Stuckey, P.~J.}
\newblock Transforming normal logic programs to constraint logic programs.
\newblock {\em TCS}, 105(1):27--56 1992.

\bibitem[Lagniez and Marquis, 2017]{LM2017}
{\sc Lagniez, J.-M.} {\sc and} {\sc Marquis, P.}
\newblock An improved decision-{DNNF} compiler.
\newblock In {\em IJCAI} 2017, volume~17, pp. 667--673.

\bibitem[Lagniez and Marquis, 2019]{LM2019}
{\sc Lagniez, J.-M.} {\sc and} {\sc Marquis, P.}
\newblock A recursive algorithm for projected model counting.
\newblock In {\em AAAI} 2019, volume~33, pp. 1536--1543.

\bibitem[Lee and Lifschitz, 2003]{LL2003}
{\sc Lee, J.} {\sc and} {\sc Lifschitz, V.}
\newblock Loop formulas for disjunctive logic programs.
\newblock In {\em ICLP} 2003, pp. 451--465. Springer.

\bibitem[Lee et~al., 2017]{LTW2017}
{\sc Lee, J.}, {\sc Talsania, S.}, {\sc and} {\sc Wang, Y.}
\newblock Computing {LPMLN} using {ASP} and {MLN} solvers.
\newblock {\em Theory and Practice of Logic Programming}, 17(5-6):942--960
  2017.

\bibitem[Leone and Ricca, 2015]{LR2015}
{\sc Leone, N.} {\sc and} {\sc Ricca, F.}
\newblock Answer set programming: A tour from the basics to advanced
  development tools and industrial applications.
\newblock In {\em Reasoning web international summer school} 2015, pp.
  308--326. Springer.

\bibitem[Lierler, 2005]{Lierler2005}
{\sc Lierler, Y.}
\newblock {Cmodels}--{SAT}-based disjunctive answer set solver.
\newblock In {\em LPNMR} 2005, pp. 447--451. Springer.

\bibitem[Lifschitz, 2010]{Lifschitz2010}
{\sc Lifschitz, V.}
\newblock Thirteen definitions of a stable model.
\newblock {\em Fields of logic and computation}, pp. 488--503 2010.

\bibitem[Lifschitz and Razborov, 2006]{LR2006}
{\sc Lifschitz, V.} {\sc and} {\sc Razborov, A.}
\newblock Why are there so many loop formulas?
\newblock {\em TOCL}, 7(2):261--268 2006.

\bibitem[Marek and Truszczy{\'n}ski, 1999]{MT1999}
{\sc Marek, V.~W.} {\sc and} {\sc Truszczy{\'n}ski, M.}
\newblock Stable models and an alternative logic programming paradigm.
\newblock In {\em The Logic Programming Paradigm} 1999, pp. 375--398. Springer.

\bibitem[Marek and Truszczy{\'n}ski, 1991]{MT1991}
{\sc Marek, W.} {\sc and} {\sc Truszczy{\'n}ski, M.}
\newblock Autoepistemic logic.
\newblock {\em Journal of the ACM (JACM)}, 38(3):587--618 1991.

\bibitem[Nogueira et~al., 2001]{NBGWB2001}
{\sc Nogueira, M.}, {\sc Balduccini, M.}, {\sc Gelfond, M.}, {\sc Watson, R.},
  {\sc and} {\sc Barry, M.}
\newblock An {A}-{P}rolog decision support system for the space shuttle.
\newblock In {\em PADL} 2001, pp. 169--183. Springer.

\bibitem[Pontelli et~al., 2009]{PSE2009}
{\sc Pontelli, E.}, {\sc Son, T.~C.}, {\sc and} {\sc Elkhatib, O.}
\newblock Justifications for logic programs under answer set semantics.
\newblock {\em TPLP}, 9(1):1--56 2009.

\bibitem[Rusovac et~al., 2024]{RHGGF2024}
{\sc Rusovac, D.}, {\sc Hecher, M.}, {\sc Gebser, M.}, {\sc Gaggl, S.~A.}, {\sc
  and} {\sc Fichte, J.~K.}
\newblock Navigating and querying answer sets: how hard is it really and why?
\newblock In {\em KR} 2024, volume~21, pp. 642--653.

\bibitem[Sharma et~al., 2019]{SRSM2019}
{\sc Sharma, S.}, {\sc Roy, S.}, {\sc Soos, M.}, {\sc and} {\sc Meel, K.~S.}
\newblock {GANAK}: A scalable probabilistic exact model counter.
\newblock In {\em IJCAI} 2019, volume~19, pp. 1169--1176.

\bibitem[Shukla et~al., 2022]{SMKS2022}
{\sc Shukla, A.}, {\sc M{\"o}hle, S.}, {\sc Kauers, M.}, {\sc and} {\sc Seidl,
  M.}
\newblock Outercount: A first-level solution-counter for quantified boolean
  formulas.
\newblock In {\em CICM} 2022, pp. 272--284. Springer.

\bibitem[Suzuki et~al., 2017]{SHS2017}
{\sc Suzuki, R.}, {\sc Hashimoto, K.}, {\sc and} {\sc Sakai, M.} 2017.
\newblock Improvement of projected model-counting solver with component
  decomposition using sat solving in components.
\newblock Technical report, JSAI.

\bibitem[Trinh et~al., 2024]{TBPS2024}
{\sc Trinh, G.}, {\sc Benhamou, B.}, {\sc Pastva, S.}, {\sc and} {\sc Soliman,
  S.}
\newblock Scalable enumeration of trap spaces in boolean networks via answer
  set programming.
\newblock In {\em AAAI} 2024, volume~38, pp. 10714--10722.

\bibitem[Valiant, 1979]{valiant1979}
{\sc Valiant, L.~G.}
\newblock The complexity of enumeration and reliability problems.
\newblock {\em SIAM Journal on Computing}, 8(3):410--421 1979.

\bibitem[Zhou, 2014]{Zhou2014}
{\sc Zhou, Y.}
\newblock From disjunctive to normal logic programs via unfolding and shifting.
\newblock In {\em ECAI 2014} 2014, pp. 1139--1140. IOS Press.

\end{thebibliography}

\appendix
\clearpage
\section*{Appendix}

\section{An Illustrative Example}
\paragraph{Example to illustrate Justification of atoms}
\addtocounter{example}{-1}
\begin{example}[continued]
    Consider the following two interpretations over $\at{P}$:
    \begin{itemize}
        \item $M_1 = \{p_0, w, q_0, q_1\}$: clearly, $\tau_{M_1} = \{p_0, w, q_0, q_1, \neg{p_1}\}$. As no proper subset of $M_1$ satisfies $P^{M_1}$, each atom of $M_1$ is justified. 
        \item $M_2 = \{p_1, w, q_0, q_1\}$: here, $\tau_{M_2} = \{p_1, w, q_0, q_1, \neg{p_0}\}$. Note that $\tau_{M_2} \models \completion{P}$. 
        The program $P^{M_2}$ includes all rules of $P$ except rule $r_7$. 
        There exists an interpretation $\{p_1, q_0\} \subset M_2$ such that $\{p_1, q_0\}$ satisfies $P^{M_2}$.  
        It indicates that atoms $q_1$ and $w$ are not justified in $M_2$. 
    \end{itemize}
\end{example}

\paragraph{Example to illustrate~\Cref{lemma:cyclic_atom_suffices}}
\addtocounter{example}{-1}
\begin{example}[continued]
    Consider the following two interpretations over $\at{P}$:
    \begin{itemize}
        \item $M_1 = \{p_0, w, q_0, q_1\}$: clearly, $\tau_{M_1} = \{p_0, w, q_0, q_1, \neg{p_1}\}$. Note that $M_1 \in \answer{P}$, as no strict subset of $M_1$ satisfies $P^{M_1}$.
        \item $M_2 = \{p_1, w, q_0, q_1\}$: here, $\tau_{M_2} = \{p_1, w, q_0, q_1, \neg{p_0}\}$. While $\tau_{M_2} \models \completion{P}$, it can be shown that $M_2 \not \in \answer{P}$. 
        The program $P^{M_2}$ includes all rules of $P$ except rule $r_7$. 
        There exists an interpretation $\{p_1, q_0\} \subset M_2$ such that $\{p_1, q_0\}$ satisfies $P^{M_2}$.  This means that the atoms $q_1$ and $w$ in $M_2$ are not justified. 
        Note that both $q_1$ and $w$ are loop atoms of program $P$. 
    \end{itemize}
\end{example}

\paragraph{Example to illustrate}$\copyop{P}$
\addtocounter{example}{-1}
\begin{example}[continued]
    For the given program $P$, we have $\loopatoms{P} = \{q_1, w\}$. 
    Therefore, $\copyop{P}$ introduces two fresh copy variables $\copyatom{q_1}$ and $\copyatom{w}$, and $\copyop{P}$ consists of 
    the following implications: $\{\copyatom{q_1} \longrightarrow q_1, \copyatom{w} \longrightarrow w, q_0 \vee \copyatom{q_1}, \copyatom{w} \longrightarrow \copyatom{q_1}, p_0 \longrightarrow \copyatom{w}, p_1 \wedge \copyatom{q_1} \longrightarrow \copyatom{w}\}$.
\end{example}

\paragraph{Example to illustrate~\Cref{prop:definition}}
\addtocounter{example}{-1}
\begin{example}[continued]
    Consider two interpretations $M_1, M_2 \subseteq \at{P}$:
    \begin{itemize}
        \item $M_1 = \{p_0, w, q_0, q_1\}$, where $\tau_{M_1} = \{p_0, w, q_0, q_1, \neg{p_1}\}$. Note that $M_1 \in \answer{P}$
        and we can verify that $\up{\copyop{P}}{\tau_{M_1}} \wedge (\neg{\copyatom{q_1}} \vee \neg{\copyatom{w}})$ is UNSAT.
        \item $M_2 = \{p_1, w, q_0, q_1\}$, where $\tau_{M_2} = \{p_1, w, q_0, q_1, \neg{p_0}\}$. Here, $M_2 \not\in \answer{P}$. While $\tau_{M_2} \models \completion{P}$, we can see that 
        $\up{\copyop{P}}{\tau_{M_2}} \wedge (\neg{\copyatom{q_1}} \vee \neg{\copyatom{w}}) = (\neg{\copyatom{w}} \vee \copyatom{q_1}) \wedge (\neg{\copyatom{q_1}} \vee \copyatom{w}) \wedge (\neg{\copyatom{w}} \vee \neg{\copyatom{q_1}})$ is SAT.
    \end{itemize}
\end{example}

\section{Proofs Deferred to the Appendix}

\cyclicatomsuffices*
\begin{proof}
    For notational clarity, let $A$ and $B$ denote the formulas $P^{M}
    \wedge \bigwedge_{x \in \tau_{M}^{-}} \neg{x} \wedge \bigvee_{x
      \in \tau_{M}^{+}} \neg{x}$ and $P^{M} \wedge \bigwedge_{x \in
      \tau_{M}^{-}} \neg{x} \wedge \bigwedge_{x \in \tau_{M}^{+}
      \wedge x \not \in \loopatoms{P}} x \wedge \bigvee_{x \in
      \tau_{M}^{+} \wedge x \in \loopatoms{P}} \neg{x}$, respectively.

    We use proof by contradiction. Suppose, if possible, $A$ is SAT
    but $B$ is UNSAT.

    Given that $A$ is SAT, we know that some atoms are not justified
    in $M$ (\Cref{prop:justification_all_atoms}).  Similarly, since
    $B$ is UNSAT, we know that all loop atoms are justified in $M$
    (\Cref{prop:justification_loop_atoms}).  Therefore, there must be
    a non-loop atom, say $x_1$, that is not justified in $M$.  Since
    $x_1 \in M$ and $\tau_{M} \models \completion{P}$, according to group $3$
    implications in the definition of Clark completion, there exists a
    rule $r_1 \in P$ such that $\body{r_1} \wedge \bigwedge_{x \in
      \head{r_1} \setminus \{x_1\}} \neg{x}$ is \true under
    $\tau_{M}$.
    It follows that there exists an atom $x_2 \in \body{r_1}^+$ that
    is not justified; otherwise, the atom $x_1$ would have no other
    option but be justified.  Now, we can repeat the same argument we
    presented above for $x_1$, but in the context of the non-justified
    atom $x_2$ in $M$.  By continuing this argument, we obtain a
    sequence of not justified atoms $\{x_1, x_2, \ldots\}$, such that
    the underlying set is a subset of $M$.  There are two possible
    cases to consider: either~(i) the sequence $\{x_i\}$ is unbounded,
    or (ii)~for some $i < j$, $x_i = x_j$. Case~(i)~contradicts the
    finiteness of $\at{P}$.  Case (ii) implies that some loop atoms are
    not justified -- a contradiction of our premise!

\end{proof}
\paragraph{Proof of~\Cref{lemma:copy_reduct_are_equal}}
\copyreductareequal*

\begin{proof}
    \textbf{The proof of ``part 1''}:\\
    Recall that $P^M$ can be thought of as a set of implications (see~\Cref{section:preliminaries}).
    First we discuss which implications are left in the subformula $P^M \wedge \bigwedge_{x \in \tau_{M}^{-}} \neg{x} \wedge \bigwedge_{x \in \tau_{M}^{+} \wedge x \not \in \loopatoms{P}} x$, following unit propagation.

    Note that the subformula $\bigwedge_{x \in \tau_{M}^{-}} \neg{x} \wedge \bigwedge_{x \in \tau_{M}^{+} \wedge x \not \in \loopatoms{P}} x$ includes a unit clause $\neg{x}$, 
    for each atom $x \in \tau_{M}^-$ and a unit clause $x$, each atom $x \in \tau_{M}^+ \cap \loopatoms{P}$.
    Since $\tau_M \models \completion{P}$, the unit propagation does not result in any empty clause. 
    When the unit propagation reaches the {\em fixed point}, all atoms that are assigned to \false ($x \in \tau_{M}^-$) and non loop atoms that are assigned to \true ($x \in \tau_{M}^{+} \wedge x \not \in \loopatoms{P}$) are removed from implications of $P^M$. 
    Thus, after reaching the fixed point, each implication is of the form: $\beta \longrightarrow \alpha$, 
    where $\beta$ (antecedent) is either \true or a conjunction of loop atoms, and $\alpha$ (consequent) is a disjunction of loop atoms, where those loop atoms are assigned to \true in $\tau_M$.
    Note that no implication left such that its consequent ($\alpha$) part is \false~but the antecedent ($\beta$) part is a conjunction of loop atoms; otherwise, we can say that $\tau_M \not \models \completion{P}$.
    In summary, after reaching the unit propagation fixed point, the formula $P^M \bigwedge_{x \in \tau_{M}^{-}} \neg{x} \wedge \bigwedge_{x \in \tau_{M}^{+} \wedge x \not \in \loopatoms{P}} x$ is left with a set of unit clauses and implications, where the implications consists of only loop atoms that are assigned to \true in $\tau_M$.

    Now we discuss how the implications of $\up{\copyop{P}}{\tau_M}$ relate to the implications of $P^M \wedge \bigwedge_{x \in \tau_{M}^{-}} \neg{x} \wedge \bigwedge_{x \in \tau_{M}^{+} \wedge x \not \in \loopatoms{P}} x$.
    According to the definition of the copy operation, $\copyop{P}$ introduces a type~\ref{l1:type2} implication for every rule whose head has at least one loop atom; 
    that is, $\copyop{P}$ introduces a type~\ref{l1:type2} implication for a rule $r$ if $\head{r} \cap \loopatoms{P} \neq \emptyset$.
    Recall that these type~\ref{l1:type2} implications are similar to the implications introduced and discussed above for $P^M$, except that  
    each of the loop atoms (i.e., $x$) of rule $r$ is replaced by its corresponding copy atom (i.e., $\copyatom{x}$). Clearly, no loop atom of $P$ presents in $\copyop{P}$.

    Note that the assignment $\tau_M$ does not assign any copy variables. 
    Since $\tau_M \models \completion{P}$, no empty clause is introduced in the unit propagation of $\up{\copyop{P}}{\tau_M}$. 
    If a loop atom $x \in \loopatoms{P}$ is \false in $\tau_M$, then the corresponding copy atom $\copyatom{x}$ will be unit propagated to \false in $\up{\copyop{P}}{\tau_M}$, due to the type~\ref{l1:type1} implication.
    Finally, no type~\ref{l1:type1} implication left in $\up{\copyop{P}}{\tau_M}$ after reaching the unit propagation fixed point.
    Straightforwardly, the formula $\up{\copyop{P}}{\tau_M}$ is the conjuction of a set of unit clauses and implications, where these implications consist of only copy atoms such that their corresponding loop atoms are assigned to \true in $\tau_M$. 
    
    From the discussion so far, we can say that the implications of $P^{M} \wedge \bigwedge_{x \in \tau_{M}^{-}} \neg{x} \wedge \bigwedge_{x \in \tau_{M}^{+} \wedge x \not \in \loopatoms{P}} x$ left after unit propagation are identical  to 
    the implications left from $\up{\copyop{P}}{\tau_M}$ after unit propagation, except that each loop atom in $P^M$ is replaced by their corresponding copy atom in $\up{\copyop{P}}{\tau_M}$. 
    Clearly, the satisfiability of both formulas $\up{\copyop{P}}{\tau_{M}} \wedge \bigvee_{x \in \tau_{M}^{+} \wedge x \in \loopatoms{P}} \neg{\copyatom{x}}$ and $P^{M} \wedge \bigwedge_{x \in \tau_{M}^{-}} \neg{x} \wedge \bigwedge_{x \in \tau_{M}^{+} \wedge x \not \in \loopatoms{P}} x \wedge \bigvee_{x \in \tau_{M}^{+} \wedge x \in \loopatoms{P}} \neg{x}$ depends only on these implications left so far; since the variables of unit clauses are different from the variables of those implications.
    From the relationship discussed above and following the~\Cref{prop:justification_loop_atoms}, the satisfiability checking of $\up{\copyop{P}}{\tau_{M}} \wedge \bigvee_{x \in \tau_{M}^{+} \wedge x \in \loopatoms{P}} \neg{\copyatom{x}}$ can be rephrased as to check justification of all loop atoms in the interpretation $\tau_M$.
    Finally we can say that the formula $\up{\copyop{P}}{\tau_{M}} \wedge \bigvee_{x \in \tau_{M}^{+} \wedge x \in \loopatoms{P}} \neg{\copyatom{x}}$ is SAT if and only if $P^{M} \wedge \bigwedge_{x \in \tau_{M}^{-}} \neg{x} \wedge \bigwedge_{x \in \tau_{M}^{+} \wedge x \not \in \loopatoms{P}} x \wedge \bigvee_{x \in \tau_{M}^{+} \wedge x \in \loopatoms{P}} \neg{x}$ is SAT (the satisfiability of both formulas have the same meaning).
    
    \noindent \textbf{The proof of ``part 2''}:\\
    We use the relationship established between the implications in $\up{\copyop{P}}{\tau_M}$ and implications in $P^M$ in \textbf{the proof of part 1}.
    
    \noindent \textbf{proof of ``if part''}: The proof of ``if part'' follows~\Cref{lemma:cyclic_atom_suffices} --- the~\Cref{lemma:cyclic_atom_suffices} proves that if some atoms of $M$ are not justified (or $P^{M} \wedge \bigwedge_{x \in \tau_{M}^{-}} \neg{x} \wedge \bigvee_{x \in \tau_{M}^{+}} \neg{x}$ is SAT), then some loop atoms of $M$ are not justified (or $\up{\copyop{P}}{\tau_{M}} \wedge \bigvee_{x \in \tau_{M}^{+} \wedge x \in \loopatoms{P}} \neg{\copyatom{x}}$ is SAT). 
    We have already shown that (in \textbf{part 1}) the satisfiability checking of $\up{\copyop{P}}{\tau_{M}} \wedge \bigvee_{x \in \tau_{M}^{+} \wedge x \in \loopatoms{P}} \neg{\copyatom{x}}$ can be rephrased as to check justification of all loop atoms in the interpretation $\tau_M$.
    So, the ``if part'' is proved.
    
    \noindent \textbf{proof of ``only if part''}: The proof is trivial. If some loop atoms of $M$ are not justified, then some atoms of $M$ are not justified. 
    So, the ``only if part'' is proved.
\end{proof}

\section{Further Experimental Analysis}
\label{section:detailed_experimental_analysis}

\paragraph{Performance comparison of different ASP counters across different computation problems.}
We present the~\Cref{table:performance_on_different_problems} showing the number of instances across different benchmark classes solved by different ASP counters.
We observe that there are two benchmark classes: preferred extension and diagnosis, where the performance of \toolname~is surpassed by \clingo. 
Our observations reveal that these instances tend to have a significantly larger number of loop atoms.
More specifically, around $66\%$ of Preferred extension instances and $100\%$ of Diagnosis instances contain more than $1000$ loop atoms.
\begin{table}[h]
    \centering
    \begin{tabular}{m{6em} m{3em} m{4em} m{4em} m{4em} m{3em} m{6em}} 
    \toprule
     & $\sum$ & $\sum^{\geq 1000}$ & clingo & DynASP & Wasp & \toolname\\
    \midrule
    $2$QBF & 200 & 0 & 179 & 0 & 58 & \textbf{181}\\
    \midrule
    Strategic & 226 & 0 & 53 & 0 & 0 & \textbf{125}\\
    \midrule
    Preferred & 217 & 142 & \textbf{208} & 2 & 192 & 110\\
    \midrule
    PC config & 1 & 1 & 0 & 0 & 0 & 0\\
    \midrule
    Diagnosis & 11 & 11 & \textbf{11} & 0 & 8 & 6\\
    \midrule
    Random & 226 & 0 & 80 & 0 & 0 & \textbf{213}\\
    \midrule
    MTS & 244 & 53 & 177 & 87 & 174 & \textbf{190}\\
    \midrule
    \midrule
    & & & 708 & 89 & 432 & 825\\
    \bottomrule
    \end{tabular}
    \caption{The performance comparison of different ASP counters across different benchmark classes. 
    The second ($\sum$) and third columns ($\sum^{\geq 1000}$) represents the total number of instances and total number of instances having more that $1000$ loop atoms in each benchmark class, respectively. 
    }
    \label{table:performance_on_different_problems}
\end{table}

\paragraph{Experimentals with Alternative Projected Model Counters ($\#\exists$SAT solvers) and ApproxASP.}
We conducted experiments with alternative projected model counters, including D$4$~\citep{LM2017}, GPMC~\citep{SHS2017}, as well as the approximate answer set counter ApproxASP~\citep{KESHFM2022}. 
The results of these experiments, comparing the performance of the alternative counting techniques, are summarized in~\Cref{table:alternative_counter_results}. 
The results reveal that \toolname~with \ganak~outperforms \toolname~with alternative projected model counters. 
\begin{table}[h]
    \centering
    \begin{tabular}{m{7em} m{3em} m{4em} m{5em}} 
    \toprule
    & D4 & GPMC & ApproxASP \\
    \midrule
    \#Solved {\small ($1125$)} & 697 & 759 & 715 \\
    \midrule
    PAR$2$ & 3856 & 3463 & 3829 \\
    \bottomrule
    \end{tabular}
    \caption{The performance comparison of alternative $\#\exists$SAT counting techniques and approximate counter ApproxASP.}
    \label{table:alternative_counter_results}
  \end{table}

\paragraph{Experimentals with $16$ GB Memory Limit}
We set the memory limit of $8$ GB, which is consistent with prior works on answer set counting~\citep{KCM2024,FGHR2024,EHK2021}.
Additionally, we conducted another set of experiments with a $16$ GB memory limit and the result is summarized in~\Cref{table:16gb_experimental_result}.
The findings indicate a slight increase in the number of solved instances for most ASP counters, with the exception of Wasp.
Specifically, the increase in solved instances was \clingo ($+10$), DynASP ($+8$), Wasp ($+96$), and \toolname~remained unchanged.
\begin{table}[h]
    \centering
    \begin{tabular}{m{7em} m{3em} m{3em} m{3em} m{6em}} 
    \toprule
    & \clingo & DynASP & Wasp & \toolname\\
    \midrule
    \#Solved {\small ($1125$)} & 718 & 97 & 528 & \textbf{825}\\
    \midrule
    PAR$2$ & 4025 & 9137 & 5552 & \textbf{2876}\\
    \bottomrule
    \end{tabular}
    \caption{The performance of \toolname~vis-a-vis existing disjunctive answer set counters, with a $16$ GB memory limit.}
    \label{table:16gb_experimental_result}
\end{table}

\end{document}